\documentclass[runningheads]{llncs}

\let\proof\relax

\usepackage{graphicx}
\usepackage{cite}
\usepackage{algorithmic}
\usepackage{graphicx}
\usepackage{textcomp}
\usepackage{xcolor}
\usepackage[ruled,vlined]{algorithm2e}
\usepackage{amsfonts}
\usepackage{bbm}
\usepackage{mathtools}
\usepackage{amsthm}
\usepackage{thmtools} 
\usepackage{thm-restate}
\usepackage{caption}
\usepackage{subcaption}
\usepackage{diagbox}
\usepackage{array}
\usepackage{amsmath}
\usepackage{amssymb}
\usepackage{mathtools}
\usepackage[english]{babel}

\usepackage{breqn}
\usepackage{enumitem}
\usepackage{physics}
\usepackage[english]{babel}
\theoremstyle{definition}
\usepackage[title]{appendix}
\theoremstyle{remark}
\DeclareMathOperator*{\argmax}{arg\,max}

\newcommand{\comment}[1]{}

\begin{document}
\title{An Incentive Mechanism for Trading Personal Data in Data Markets}
\titlerunning{An Incentive Mechanism for Data Markets}
%
\author{Sayan Biswas\inst{} \and
Kangsoo Jung\inst{} \and
Catuscia Palamidessi\inst{}}
\authorrunning{}
%
\institute{Inria and \'{E}cole Polytechnique, France\\
\email{\{sayan.biswas,gangsoo.zeong\}@inria.fr}\\
\email{catuscia@lix.polytechnique.fr}
\url{}}

%
\maketitle              
\begin{abstract}
With the proliferation of the digital data economy, digital data is considered as the crude oil in the twenty-first century, and its value is increasing. Keeping pace with this trend, the model of data market trading between data providers and data consumers, is starting to emerge as a process to obtain high-quality personal information in exchange for some compensation. However, the risk of privacy violations caused by personal data analysis hinders data providers' participation in the data market. Differential privacy, a de-facto standard for privacy protection, can solve this problem, but, on the other hand, it deteriorates the data utility. In this paper, we introduce a pricing mechanism that takes into account the trade-off between privacy and accuracy. We propose a method to induce the data provider to accurately report her privacy price and, we optimize it in order to maximize the data consumer's profit within budget constraints. We show formally that the proposed mechanism achieves these properties, and also, validate them experimentally.

\end{abstract}

\section{Introduction}
Nowadays, digital data is becoming an essential resource for the information society, and the value of personal data is increasing. In the past, data broker companies such as Acxiom collected personal data and sold them to  companies that needed them. However, as the value of personal data is becoming 
clear to the data providers, and concern about their privacy is increasing among them, people are less and less willing to let their data to be collected for free. In this scenario, the model of \emph{data market} is starting to emerge, as a process to obtain high-quality personal information in exchange of a compensation. Liveen \cite{Liveen} and Datacoup \cite{Datacoup} are examples of prototypes  of  data market services, where the data providers can obtain additional revenue from selling their data, and the consumers can collect the desired personal data.

The problem of privacy violation by personal data analysis is one of the major issues in such data markets. As the population becomes more and more aware of the negative consequences of privacy breaches, such as the Cambridge Analytica scandal, people are reluctant to release their data, unless they are properly sanitised. In order to solve this problem, techniques like noise insertion \cite{dwork2014}, synthetic data \cite{bowen2019}, secure multi-party computation (SMC) \cite{volgushev2019}, and homomorphic encryption \cite{acar2018} are being actively studied. Differential privacy \cite{dwork2014}, a de-facto standard for privacy protection, is one of the techniques to prevent privacy violations in the data market. 

Differential privacy provides a privacy protection framework based on solid mathematical foundations,  and enables quantified privacy protection according to the amount of noise insertion. However, like all privacy-protection methods, it deteriorates the data utility. If the data provider inserts too much noise because of privacy concern, the data consumer cannot proceed with the data analysis with the required  performance. This trade-off between privacy and utility is a long-standing problem in differential privacy. The privacy protection and data utility depend on the amount of noise insertion while applying differential privacy, and the amount of noise insertion is determined by the noise parameter $\epsilon$. Thus, determining the appropriate value of the parameter $\epsilon$ is a fundamental problem in differential privacy. It is difficult to establish the appropriate $\epsilon$ value because it depends on many factors that are difficult to quantify, like the attitude towards privacy of the data provider, which may be different from person to person.

We propose an incentive mechanism to  encourage the data providers to join in the data market  and motivate them to share more accurate data. The amount of noise insertion depends on the data providers' privacy preference and the incentives provided to them by data consumers, and the data consumers decide on incentives to pay to the data provider by considering the profit to be made from the collected data. By sharing some of the consumers' profit with the data provider as incentive, the data provider can get fair prices for providing her data. 
The proposed mechanism consists of the truthful price report mechanism and an optimization method within budget constraints. The truthful price report mechanism guarantees that the data provider takes the optimal profit when she reports her privacy price to the data consumer honestly. Based on a data provider’s reported privacy price, a data consumer can maximize her profit within a potential budget constraint. 

\subsection{Contribution}
The contributions of this paper are as follows:
\begin{enumerate}[label=(\roman*)]
\item Truthful price report mechanism: We propose an incentive mechanism that guarantees that the data provider maximizes her benefit when she reports her privacy price honestly.
\item Optimized incentive mechanism within the budget constraints: We propose an optimization method to maximize the data consumer’s profit and information gain in the setting where the data consumer has a fixed financial budget for data collection.
\item Optimized privacy budget splitting mechanism:  We propose a method of splitting the privacy budget for the data providers, that allows them to maximize her utility-gain within a fixed privacy budget, in a multiple data consumer environment.
\end{enumerate}
The properties of our methods are both proved formally and validated through experiments. 

\subsection{Structure of the paper}
The structure of this paper is as follows: we explain the related works and preliminaries in Sections 2 and 3, respectively. We describe the proposed incentive mechanism in Section 4 and validate the proposed incentive mechanism through experiments in Section 5. Our conclusion and some potential directions of future work are discussed in Section 6.

\section{Related work}
\subsection{Methods for choosing $\epsilon$}
In differential privacy concept, parameter $\epsilon$ is the knob to control the privacy-utility trade off. The smaller the $\epsilon$, the higher is the privacy protection level and the more it deteriorates the data utility. Conversely, a larger $\epsilon$ decreases the privacy protection level and enhances the data utility. However, there is no gold standard to determine the appropriate value of $\epsilon$. Apple has been promoting the use of differential privacy to protect user data since iOS 10 was released, but the analysis of \cite{tang2017} showed the $\epsilon$ value was set at approximately 10 without any particular reason. The work of \cite{lee2011} showed that the privacy protection level set by an arbitrary $\epsilon$ can be infringed by inference using previously disclosed information and proposed an $\epsilon$ setting method considering posterior probability. This matter is the main factor that undermines the claim that personal information is protected by differential privacy. Much research have been conducted to study and solve this problem \cite{chen2016,ligett2012,xiao2013,nissim2012}. Although a lot of research is being done in this area, the problem of determining a reasonable way of choosing an optimal value for $\epsilon$ still remains open, as there are many factors to consider in deciding the value of $\epsilon$, and more studies are still needed. In this paper, we propose a technique to determine an appropriate value of $\epsilon$ by setting a price of the privacy of the data provider.

\subsection{Pricing mechanism}
One of the solutions to find an appropriate value of $\epsilon$ is to price it according to the data accuracy \cite{hsu2014,ghosh2015,dandekar2012,roth2012,fleishcer2012,li2018}. In \cite{hsu2014}, strength of the privacy guarantee and the accuracy of the published results are considered to set the $\epsilon$ value, and a simple $\epsilon$ setting model that can satisfy data providers and consumers was suggested. In \cite{ghosh2015}, the author proposed a compensation mechanism via auction in which data providers are rewarded based on data accuracy and data consumer's budget when they provide data with differential privacy. It is the most similar work to our study. The main differences between our paper and Ghosh and Roth's work are as follows:
\begin{enumerate}[label=(\roman*)]
    \item We define a truthful price report mechanism that a data provider get a best profit when she reports her privacy price honestly, and prove it.
    \item We propose an optimized incentive mechanism to maximize the data consumer's profit with a fixed expense budget, and a privacy budget splitting method to maximize the data provider's utility-gain in a multi-data consumer environment. 
\end{enumerate}

In \cite{fleishcer2012} the authors design a mechanism that can estimate statistics accurately without compromising the user’s privacy. They propose a Bayesian incentive and privacy-preserving mechanism that guarantees privacy and data accuracy.  The study of \cite{li2018} proposes a Stackelberg game to maximize mobile users who provides their trajectory data. 

Several techniques for pricing data assuming a data market environment have been studied in \cite{nget2018,oh2019,li2013,aperjis2012,jung2019,krehbiel2019,zhang2021,jorgensen2015}.

In \cite{nget2018} the authors suggested a data pricing mechanism to make the balance between privacy and price in data market environment. In \cite{oh2019}, the authors propose the data market model in the IoT environment and show the proposed pricing model has a global optimal point. In \cite{li2013} the authors proposed a theoretical framework for determining prices to noisy query answer in the differentially private data market. However, this research cannot flexibly reflect the requirements of the data market. In the study of \cite{jung2019}, the author proposed an $\epsilon$-choosing method based on Rubinstein bargaining and assumes a market manager that mediates a data provider and consumer in the data trading.

It is realistic to consider personal data as a digital asset, and reasonable to attempt to find a bridge between privacy protection level and price according to the value of $\epsilon$ in differential privacy, as has been done in this paper. Existing studies are attempting to find an equilibrium between data providers and consumers under the assumption that both are reasonable individuals. In this paper, we follow a research direction similar to existing studies, and focus on the incentive mechanism that motivates a data provider report her privacy price honestly. In particular, we consider that the value of differentially private data increases non-linearly with respect to the increase of the value of $\epsilon$.

\section{Preliminaries}
In this section, we explain the basic concepts of differential privacy. Differential privacy is a mathematical model that guarantees the privacy protection at a specified level $\epsilon$. For all datasets $D_1$ and $D_2$ differing exactly at a single element, it is defined to satisfy $\epsilon$-differential privacy, if the probability distribution difference of the result of a specific query $K$ on two databases is less than or equal to the threshold $e^{\epsilon}$. The definition  of the differential privacy is as follows:

\begin{definition}[Differential privacy\text{\cite{dwork2014}}]
\label{dp}
A randomized function $\mathcal{K}$ provides \emph{$\epsilon$-differential privacy} if all datasets, $D_1$ and $D_2$, differing by only one element, and all subsets, $S \subseteq$ \text{Range}($\mathcal{K}$),
\begin{equation*}
\mathbb{P}[\mathcal{K}(D_1) \in S] \leq e^{\epsilon} \mathbb{P}[\mathcal{K}(D_2) \in S]
\end{equation*}
\end{definition}

The Laplace mechanism~\cite{dwork2014} is one of the most common methods for achieving the $\epsilon$-differential privacy.

One of the important properties of differential privacy is the compositionality that allows query composing to facilitate modular design \cite{dwork2014}.

\begin{description}
\item[Sequential compositionality] For any database $D$, let we query on the randomization mechanism ${K_1}$ and ${K_2}$ which is independent for each query. The results of ${K_1}(D)$ and ${K_2}(D)$ whose guarantees are the $\epsilon_1$ and $\epsilon_2$-differential privacy, is  ($\epsilon_1 + \epsilon_2$)-differentially private.
\item[Parallel compositionality] Let $A$ and $B$ be the partition of any database $D (A \cup B = D, A \cap  B =  \phi)$. Then, the result of the query on the randomization mechanism ${K_1}(A)$ and ${K_2}(B)$, is the max($\epsilon_1, \epsilon_2$)-differentially private.
\end{description}


Recently, a variant of differential privacy called \emph{local differential privacy} has been proposed \cite{erlingsson2014,cormode2018,thong2016,kasiviswanathan2011}. In this model, data providers obfuscate their own data by themselves. Local differential privacy has an advantage that it does not need a trusted third party to satisfy the differential privacy. The properties of parallel and sequential compositionality hold for the local model as well.

In the rest of this paper, we consider the local model of differential privacy. 

\section{Incentive mechanism for data markets}

\subsection{Overview of the proposed technique}\label{sec:overview}
The data market aims at collecting personal data legally with the consent of the provider. A data provider can sell her own data and get paid for it, and a data consumer can collect the personal data for analysis by paying a price, resulting in   a win-win situation. 

Naturally, the data consumer wants to collect personal data as accurately as possible at the lowest possible price, and the data provider wants to sell her data at a price as high as possible while protecting sensitive information. In general, every effective protection technique affects the utility of the data negatively. In the particular case of differential privacy, the levels of utility and privacy are determined by the parameter $\epsilon$; thus, the data price is affected directly by the value of $\epsilon$. 

Determining  the appropriate value of $\epsilon$ and the actual price of the data are critical to the success of the data market. However this is not an easy task,  also because each data provider has different privacy needs \cite{kasiviswanathan2011}. 

We propose an incentive mechanism to find the price of the data and the value of $\epsilon$ that can satisfy both the data provider and the data consumer. The proposed method consists of two parts: an incentive mechanism encouraging the data provider to report her privacy price honestly to the data consumer, and an optimization scheme to maximize both the data consumer and provider's profit within a budget constraint. 

We consider a scenario with  $n$ data providers, $u_1,\ldots,u_n$, and $m$ data consumers, $D_1,\ldots,D_m$, and where each provider and consumer proceeds with the deal independently (we use the term ``data provider'' and ``data producer'' interchangeably, in the same sense). The term ``$\epsilon$ unit price'' 
(e.g., $1$\$ per $\epsilon$ value $0.1$) will be used to express the price, where $\epsilon$ is the parameter of differential privacy, which is a measure of the accuracy of information.  We recall that, as $\epsilon$ increases, the data becomes less private and more information can be obtained from it, and vice versa. Thus, the price per unit $\epsilon$ represents the ``value'' of the provider's information\footnote{The $\epsilon$ unit price can be of any form, including a monetary one. The method we propose is independent from the nature of the price, so we do not need to specify it.}. The price of $\epsilon$ is expected to differ from one data provider to another, because each individual has a different privacy need. We denote the $\epsilon$ unit price reported by $u_i$ as $p_i$ and  her true $\epsilon$ unit price as $\pi_i$. 

\begin{figure}[htbp]
\centerline{\includegraphics[width=0.6\textwidth]{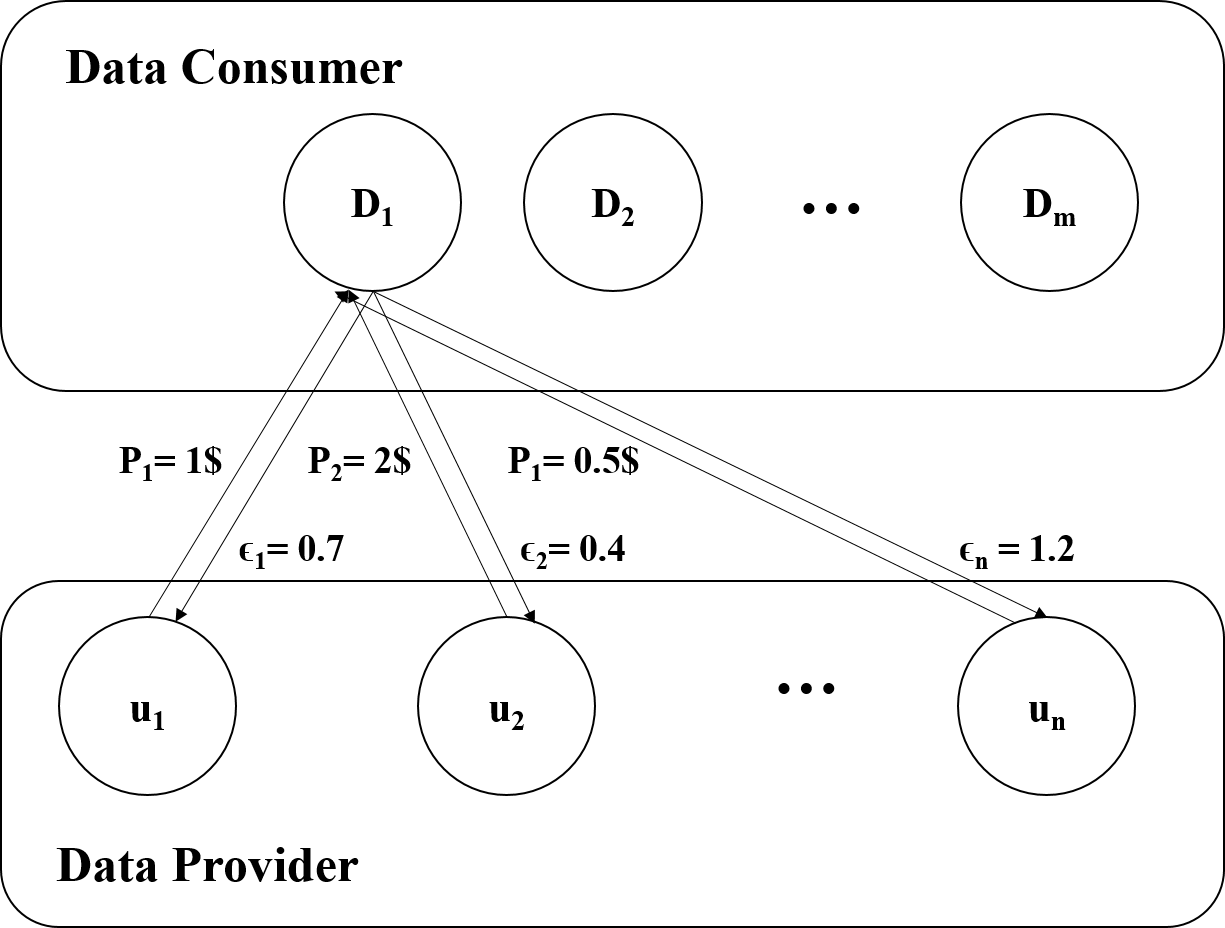}}
\caption{ An example of data trading process. In this figure, $u_i$ means the $i^{th}$ data provider and $D_j$ means the $j^{th}$ data consumer.}
\label{fig2}
\end{figure}

Fig.~\ref{fig2}  illustrates how the   process works.  At first, every data consumer broadcasts a function $f$ to the data providers, which represents the amount of data (expresses in $\epsilon$ units) the consumer  is willing to buy for a given $\epsilon$ unit price. Each  consumer has her own  such function, and it can differ from one consumer to another. We will call it \emph{$\epsilon$-allocating function}. 
We assume $f$ to be monotonically decreasing, as the consumers naturally prefer to buy more data from those data producers who are willing to offer them for less.  Note that the product  $p_i f(p_i)$ represents the total amount that will be payed by the data producer to the  consumer if they agree on the trade. 
The function  $f$ however has also a second purpose: as we shown in Section~\ref{sec:truth}, it is designed to encourage providers to demand the price that they really consider the true price of their privacy, rather than asking for more. 

Then, thanks to the truthful price report mechanism (cf. Section~\ref{sec:truth}), the data providers report the prices of their data honestly to the data consumers in accordance with the published $f$. In the example in Fig.~\ref{fig2}, $u_1$ reports her $\epsilon$ price per $0.1$ as $1$\$ and $u_2$ reports her $\epsilon$ price per $0.1$ as $2$\$. Finally, the data consumer checks the price reported by the data provider and determines the total price and value of $\epsilon$ to be obtained from each provider using $f$. In this example, the data consumer $D_1$ determines $\epsilon_1$ to be 0.7 and $\epsilon_2$ to be 0.4. 

Then, the data providers select the  consumers to whom to sell their data in order to maximize their profits, and confirm with them the values of their $\epsilon$ and the total price they would receive. In the example in Fig.~\ref{fig2}, $D_1$ pays 7\$ to $u_1$ and 8\$ to $u_2$. Finally, the data providers add noise to their data based on the determined $\epsilon$ and share the sanitized data with the respective consumers, and the  consumers pay the corresponding prices to the  providers. We assume that data providers and consumers keep the promise of the value of $\epsilon$ and 
compensation decided in the deal, once confirmed. 

This process can be repeated until the data consumers exhaust all their budget or achieve the targeted amount of information. The task of allocating a suitable budget in each round and the how to determine the amount of needed information are also important topics, but they are out of the scope of this paper and are left for future work.

\subsection{Truthful price report mechanism}\label{sec:truth}
For the correct functioning of the data trading, the data provider should be honest and demand her true privacy price. However, she may be motivated  to report a higher price, in the hope to persuade the data consumer that the information is ``more valuable'', and be willing to pay more. Note also that the true privacy price of each data provider is a personal information that only the provider herself knows and is not obliged to disclose.  

To solve this problem, we propose a truthful price report mechanism to ensure that the data providers report their $\epsilon$ unit prices  honestly. The purpose of the mechanism is to provide incentive so that the  providers are guaranteed to get the greatest profit when they report their true price.

When the data provider reports her price $p_{i}$, the data consumer determines the amount of $\epsilon$ to purchase using  $f(p_{i})$, where  $f$ is the $\epsilon$-allocating function introduced in  Section~\ref{sec:overview}.  We recall that
$f$ is a monotonically decreasing function, chosen by the consumer. 
We assume that the domain of $f$, the $\epsilon$ price unit, is normalized to  take  values in the interval  $[0,1]$.  
The total price for the data  estimated by the consumer is the product 
of the $\epsilon$ price unit and the amount to be purchased, namely, $p_{i}f(p_{i})$. To this value, the consumer adds an \emph{incentive} $\int_{p_{i}}^{\infty}f(z)\,dz$, the purpose of which is to make convenient for the data producer to report the true price (we assume that the data producer knows $f$ and the strategy of the consumer in advance). The consumer should of course choose $f$  so to be happy with the incentive. In particular, the incentive should be finite, so the contribution of $f(z)$ should vanish as  $z$ goes to ${\infty}$. An example of such a  function is illustrated in Fig.~\ref{fig3}. 

Thus the data consumer sets the \emph{offer}  $\mu(p_{i})$ to the provider $u_i$ as follows:

\begin{definition}[Payment offer]\label{incentive}
The \emph{offer} $\mu(p_{i})$ is defined  as: 
\begin{equation*}
\mu(p_{i})=p_{i}f(p_{i})+\int_{p_{i}}^{\infty}f(z)\, dz
\end{equation*}
\end{definition}

\begin{figure}[htbp]
\centerline{\includegraphics[width=0.5\textwidth]{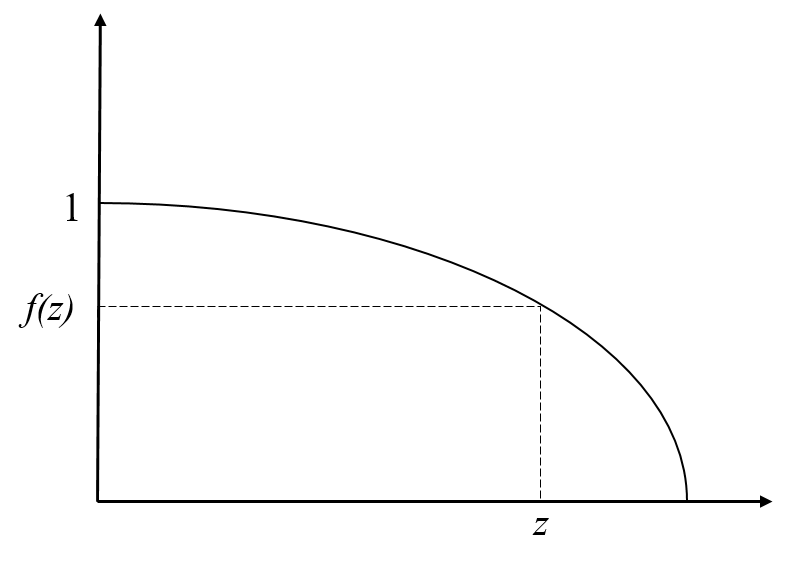}}
\caption{ An example of a monotonically decreasing function $f(z)$. Let $c$ be a parameter representing the “reported value-to-admitted $\epsilon$ value” ratio. For $z\geq 0$, we set $f(z)$ as 
$f(z)=ln(e-cz)$ if  $(e-cz)\leq 1$, and $f(x)=0$ otherwise. }
\label{fig3}
\end{figure}
We now illustrate how this strategy achieves its purpose of convincing the consumer to report her true price. We start by defining the \emph{utility} that the data provider obtains by selling her data  as the difference between the offer and the true price of her data, represented by the product of the true $\epsilon$ unit price and the amount to be sold, namely $\pi_{i} f(p_i)$: 

\begin{definition}[Utility of the data provider]
\label{utility}
The \emph{utility} $\rho(p_i)$, of the provider $u_i$, for the reported price $p_i$, is defined as:
\begin{equation*}
   \rho(p_{i})=\mu(p_{i})-\pi_{i}f(p_{i}) 
\end{equation*}
\end{definition}

We are now going to show that he proposed mechanism guarantees truthfulness. The basic reason is that each  provider $u_i$ 
achieves the best utility when reporting the true price. 
Namely,   
$\rho(\pi_{i})\geq \rho(p_{i})$ for any $p_{i}\in \mathbb{R^+}$, where we recall that $\pi_i$ is the true price of the provider $u_i$.  The only technical condition is that the function $f$ is monotonically decreasing. Under this assumption, we have the following results (see also Fig.~\ref{fig4} to get the intuition of the proof):


\begin{figure}[htbp]
\centerline{\includegraphics[width=0.6\textwidth]{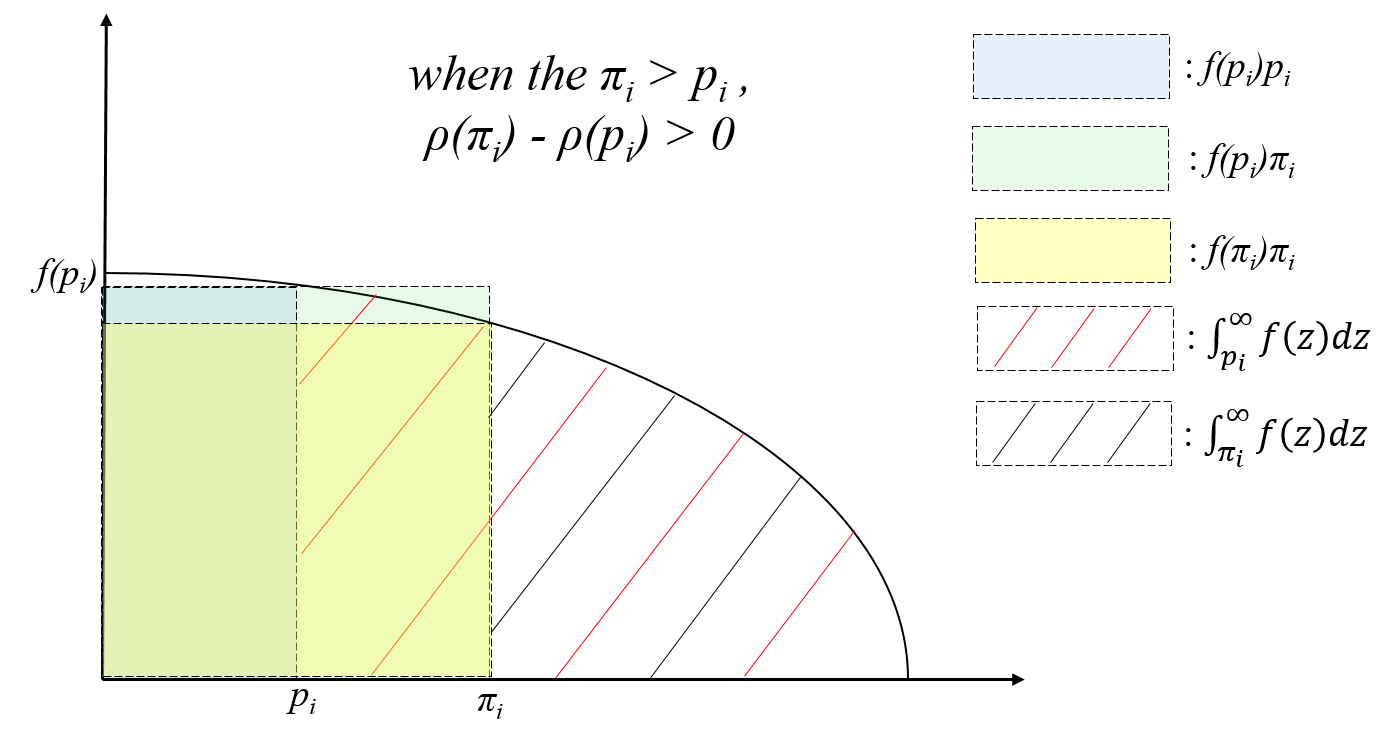}}
\caption{Graphical illustration of Theorem~\ref{mainthm}. We prove that $\rho(\pi_{i})$(blue hatching area) is always larger than $\rho(p_{i})$(blue rectangle area$+$red hatching area$-$green rectangle area).}  
\label{fig4}
\end{figure}

\begin{restatable}{lemma}{lemmagrtr}\label{reportedgrtr}
If $u_i$ reports a price greater than her true price, i.e.,  $p_i \geq \pi_i$, then her utility  will be less than the utility for the true price, i.e., $\rho(p_i)\leq \rho(\pi_i)$.
\end{restatable}
\proof See Appendix~\ref{app:a}.

\begin{restatable}{lemma}{lemmasmll}\label{reportedless}
If $u_i$ reports a price smaller than her true price, i.e.,  $p_i \leq \pi_i$, then her utility will be less than the utility for the true price, i.e., $\rho(p_i)\leq \rho(\pi_i)$.
\end{restatable}
\proof See Appendix~\ref{app:a}.\\

Combining Lemma \ref{reportedgrtr} and Lemma \ref{reportedless} gives the announced result. 
We assume of course that each data producer is a rational individual, i.e., capable of identifying the best strategy to maximize her utility.

\begin{restatable}{theorem}{thmmainthm}\label{mainthm}
If every data data producer acts rationally, then the
proposed incentive mechanism guarantees the truthfulness of the system.
\end{restatable}
\begin{proof} Immediate from Lemma \ref{reportedgrtr} and Lemma \ref{reportedless}.
\end{proof}

\subsection{Optimizing the incentive mechanism}\label{sec:optm}

In this section, we propose an optimization mechanism to identify   an optimal function $f$ 
for the data consumer with respect to the  following two desiderata:
\begin{enumerate}[label=(\roman*)]
    \item \textit{Maximum Information:} maximize the total information gain of the data consumer  with a fixed budget.   
    \item \textit{Maximum Profit:} maximize the total profit  of the data consumer with a fixed budget.    
\end{enumerate}
By ``budget'' here we mean the budget of the data consumer to  pay the data providers.

We start by introducing the notions of total information and profit for the consumer. 
Note that, \emph{by the sequential compositionality of differential privacy},  the total information   is the sum of the information obtained from each data provider.
 
\begin{definition}[Total information]\label{def:ti}
The \emph{total information} $\mathcal{I}(\vb*{u})$ obtained by the data consumer by concluding trades with each of the data providers of the tuple $\vb*{u}=(u_1,\ldots,u_n)$
 is defined as
 \[
 \mathcal{I(\vb*{u})}=\sum_{i=1}^nf(p_i)
 \]
\end{definition}

As for the profit, we can reasonably assume  to be monotonically increasing with the amount of information obtained, and that the total profit  is the sum of the profits obtained with each individual trading. The latter is naturally defined as the difference between the benefit (aka \emph{payoff}) obtained by re-selling or processing the data, and the price payed to the data provider. 

\begin{definition}[Payoff and profit]
\begin{itemize}[label=$\bullet$]
\item The payoff function for the data consumer, denoted by  $\tau(\cdot)$, is  the benefit that the data consumer receives by processing or selling the information gathered from the different data providers. The argument of $\tau(\cdot)$ is $\epsilon$, the amount of the information received. We assume $\tau(\epsilon)$ to be monotonically increasing with $\epsilon$. 
\item The total profit for the data consumer is given by $\sum_{i=1}^n(\tau(\epsilon_i)-\mu(\epsilon_i))$, where $\epsilon_i=f(p_i)$, i.e., the $\epsilon$-value allocated to $u_i$.
\end{itemize}
\end{definition}

We will consider a family of functions $\mathcal{F}$ indexed by a 
parameter $c$,
to which the $\epsilon$-allocating function $f$ belongs. The parameter $c$ reflects the data consumer's will to collect the information and, for technical reasons, we  assume $f$ to be continuous, differentiable and concave  with respect to it. For each data provider, different values of $c$ will give different $f$, that, in turn, will give rise to a different incentive-curve as per equation \eqref{incentive}, which the data consumer should adhere to for compensating for the information obtained from that data provider.

As described in previous sections, the $\epsilon$-allocating functions should be monotonically decreasing with the $\epsilon$ unit price, as the consumer is motivated to buy more information from the consumers that offer it at a lower price. This property also ensures, by Theorem \ref{mainthm},  that the prices reported by the data producers will be their true prices. 
Hence we impose the following constraint on $\mathcal{F}$:
\begin{equation}\label{eq:F}
\begin{array}{llll}
\mathcal{F} &\;\;\subseteq \;\;&\{f(\cdot,\cdot): &c, p \in \mathbb{R}^+ , f(c,p) \text{ is continuous, differentiable} \\ 
&&& \text{and concave on } c,
\text{and decreasing with }p \}.
\end{array}
\end{equation}
Note that we have added the parameter $c$ as an additional argument in $f$, so $f$ has now two arguments.

\begin{example}\label{exa:Flog}
An example of such class $\mathcal{F}$ is that of Fig.~\ref{fig3}:
\[\mathcal{F}= \{\ln(e-cp): c \in \mathbb{R}^+\}\,.\]
\end{example}

\begin{example}\label{exa:Flin}
Another example is:
\[\mathcal{F}= \{1-cp: c \in \mathbb{R}^+\}\,.\] 
\end{example}

After 
the prices $p_1,\ldots, p_n$ have been reported by the data producers $u_1,\ldots,u_n$, the data consumer will try to choose an optimal $c$ maximizing her profit.
Fig.~\ref{fig5} illustrates an example with two data provider's incentive graph and payoff for data consumer.

We will analyse the possibility to choose an optimal $c$, that, in turn, leads to an optimal $f(c,\cdot)$ addressing scenarios (i) and (ii).

\begin{figure}[htbp]
\centerline{\includegraphics[width=0.75\textwidth]{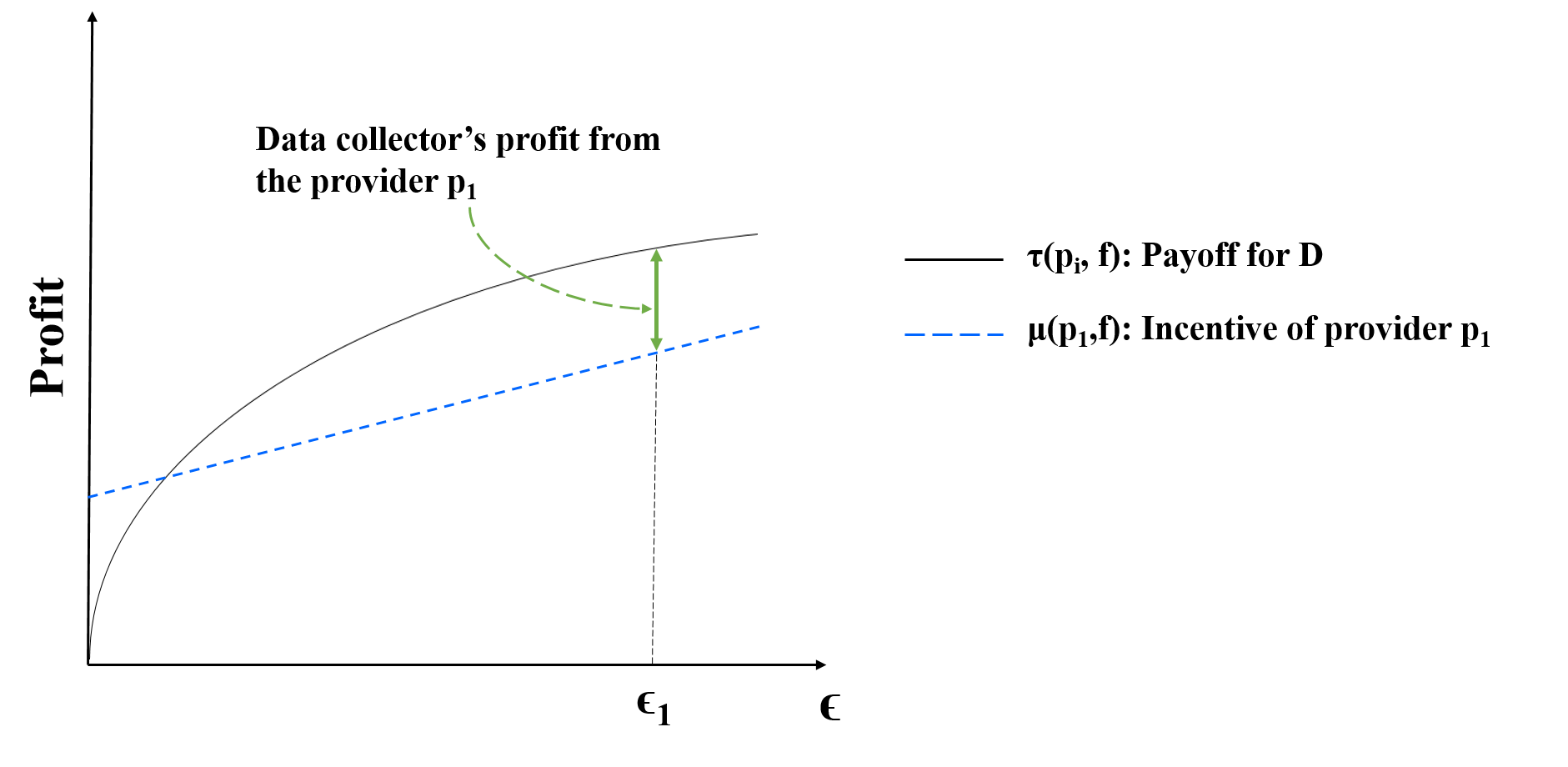}}
\caption{Illustrating the payoff for $c$ and the incentive-plots for the data consumer involving two data providers reporting $p_1, p_2$. The Y-intercept of $\mu_1$ is $\int_{p_1}^{\infty}f(z)dz$ and that for $\mu_2$ is $\int_{p_2}^{\infty}f(z)dz$.}
\label{fig5}
\end{figure}



In the context of differential privacy, we may assume that $\tau$ (the data consumer's payoff function) is additive, i.e., 
\begin{equation}
\mbox{\bf Additivity} \quad \tau(a+b)=\tau(a)+\tau(b) \quad \mbox{ for every } a,b \in \mathbb{R}^+\,.
\end{equation}
This is a reasonable assumption that goes well along with the sequential compositionality property of differential privacy, at least for small values of $\epsilon$\footnote{From a technical point of view, the additive property holds also for large values of $\epsilon$. 
 However, from a practical point of view, 
for large values of $\epsilon$, for instance $200$ and $400$,  then the original information is almost entirely revealed in both cases, and would not make sense to pay twice the price of $200$ $\epsilon$ units to achieve $400$ $\epsilon$ units.}.

We start by showing that the two desiderata (i) and (ii) are equivalent: 

\begin{restatable}{theorem}{thmequival}\label{th:two}
If $\tau(\cdot)$ is additive, then maximizing information and maximizing profit (desiderata (i) and (ii)) are equivalent, in the sense that a   $\epsilon$-allocating function $f(\cdot,\cdot)$ that maximizes the one, maximizes also the other. \label{equival}
\end{restatable}
\proof See  Appendix~\ref{app:a}.

\begin{corollary}
If $\tau(\cdot)$ is additive, then   the optimal choice of $f(\cdot,\cdot)$ w.r.t. the selected family of functions will maximize both the information gain and the profit for the data consumer. 
\end{corollary}
\begin{proof} Immediate  from Theorem \ref{equival}.
\end{proof}


We now consider the complexity problem for finding the optimal $f(\cdot,\cdot)$. 
Due to the assumptions made in Equation~\ref{eq:F}, and to the additivity of $\tau$, we can apply the method of the Lagrangians to find such $f(\cdot,\cdot)$  (cf. Appendix~\ref{app:a}).

\begin{restatable}{theorem}{maxproffixedbudget}\label{th:three}
If $\tau$ is additive,  
then there
exists a $c$ that gives an optimal {\bf profit-maximizing} function $f(c,\cdot)\in\mathcal{F}$,  for a fixed  budget, 
and we can derive such $c$ via the  method of the Lagrangians.
\label{maxproffixedbudget}
\end{restatable}
\begin{proof} See  Appendix~\ref{app:a}.
\end{proof}
\begin{restatable}{theorem}{maxinffixbudget}\label{th:four}
There exists a $c$ that gives an optimal {\bf information-maximizing} function $f(c,\cdot)\in\mathcal{F}$,  for a fixed  budget, 
and we can derive such $c$ via the  method of the Lagrangians.
\label{maxinffixbudget}
\end{restatable} 
\proof See Appendix~\ref{app:a}.\\

To demonstrate how the method works, we  show how to compute  the specific values of $c$ on the two classes $\mathcal{F}$ of Examples~\ref{exa:Flog} and \ref{exa:Flin}. Such $c$  gives the optimal $\epsilon$-allocating function $f(c, \cdot)$, maximizing $\mathcal{I}(\vb*{u})$ for a given budget.
The derivations are described in detail in Appendix~\ref{app:b}. In each example, $p_i$ is the reported {$\epsilon$ unit price}
of $u_i$.

\begin{example} \label{flog}
Let $\mathcal{F}=\{\ln(e-cp): c \in \mathbb{R}^+\}$. The optimal parameter $c$ is the solution of the equation $\ln(\prod_{i=1}^ne^{p_i}(e-cp_i)^{\frac{e}{c}})= B+\frac{n(e-1)}{c}$.
\end{example}

\begin{example} \label{flinear}
Let $\mathcal{F}=\{1-cp: c \in \mathbb{R}^+\}$. The optimal parameter $c$ is the solution of the equation $c^2\sum_{i=1}^np_i^2+2Bc-n=0$.
\end{example}

\subsection{Discussion}

In our model, for the scenario we have considered so far, the parameter $c$ is determined by the number of providers and the budget. We observe that, in both Examples \ref{flog} and \ref{flinear}, if $n$ increases than $c$ increases, and vice versa. This seems natural, because in  the families of both these example
$c$ the incentive that the consumer is going to propose decreases monotonically with $c$. This means that the larger is the offer, the smaller is the incentive that the consumer needs to be paying. 
In other words, the examples confirm the well known market law according to which the price decreases when the offer increases, and vice versa. 

We note that we have been assuming that there is enough offer to satisfy the consumer's demand. If this hypothesis is not satisfied, i.e., if the offer is smaller than the demand, then the 
situation is quite different: now the data producer can choose to whom to sell hid data. In particular, the data consumer who sets a lower $c$ will have a better chance to buy data because, naturally, the  provider  prefers to sell her data to the data consumers who give a higher incentive. 
In the next section we explore in more detail the process, from the perspective of the data provider, in the case in which the demand is  higher than the offer.  

\subsection{Optimized privacy budget splitting mechanism for data providers}

After optimizing an incentive mechanism for a given data consumer dealing with multiple data providers, we focus on the flip side of the setup. We assume a scenario in which a given data provider has to provide her data to multiple data consumers, and that there is enough demand so that she can sell all her data. 



Let there be $m$ data consumers, $D_1,\ldots,D_m$ seeking to obtain data from the user $u$. By truthful price report mechanism, as discussed in Section \ref{sec:truth}, $u$ reports her true price to each $D_i$. As discussed in Section \ref{sec:optm}, $D_i$ computes her optimal $\epsilon$-allocating function $f_i$ and requests data from $u$, differentially privatized with $\epsilon=f_i(\pi)$. After receiving $f_1,\ldots,f_m$, $u$ would like to provide her data in such a way that maximizes her utility received after sharing her data. 
\begin{definition}
We say that the data provider has made a \emph{deal} with the data consumer $D_i$ if, upon reporting the true per-unit price of her information, $\pi$, she agrees to share her data  privatized with privacy parameter $\epsilon=f_i(\pi)$.
\end{definition}

It is important to note here that $u$ is not obliged to deal with any data consumer $D_i$, even after receiving $f_i$. Realistically, $u$ has a privacy budget of $\epsilon_{\text{total}}$, which she would not exceed at any price. Let $S=\{i_1,\ldots,i_k\}$ be an arbitrary subset $\{1,\ldots,m\}$. By the sequential composition property of differential privacy, the final privacy parameter achieved by $u$ by sharing her data to an arbitrary set of data consumers $D_{i_1},\ldots,D_{i_k}$ is $\epsilon_S=\sum_{j\in S}f_{j}(\pi)$. $u$'s main intention is to share her data in such a way that ensures $\epsilon_S \leq \epsilon_{\text{total}}$ for all subset $S$ of $\{1,\ldots,n\}$, while maximizing $\sum_{j\in S}\rho_i(\pi,f_j)$, i.e., the total utility received. Reducing it down to the 0/1 knapsack problem, we propose that $u$ should be dealing with $\{D_{i_1},\ldots,D_{i_k}\}$ where $S^*=\{i_1,\ldots,i_k\}\subseteq\{1,\ldots,m\}$, chosen as 
\begin{equation}
    \resizebox{0.9\hsize}{!}{$S^*=\argmax\limits_S\{\sum\limits_{j\in S} \rho(\pi,f_j) \vert S \subseteq \{1,\ldots,m\}, \sum_{j\in S}f_j(\pi)\leq \epsilon_{\text{total}}\}$}
\end{equation}

We show the pseudocode for the $\epsilon$ allocation algorithm and the entire process in Algorithms 1 and 2.

\begin{algorithm}

\SetAlgoLined
\textbf{Input:} $\{\epsilon_1,\ldots,\epsilon_n\}$ stored in array w, $\{p_1,\ldots,p_n\}$ stored in array v, $\epsilon_{\text{total}}$   \;
\textbf{Output:} List of data consumer $\{D_1 \ldots D_k\}$ that is selected to sell data\;
 \textbf{initiate} Two-dimension array m\;
 \While{ $i \leq n $}{
 \While{ $j \leq \epsilon_{\text{total}}$}
 {
  \eIf{$w[i] > \epsilon_{\text{total}}$}{
     m[i, j] := m[i-1, j]
  }{
    m[i, j] := max(m[i-1, j], m[i-1, j-w[i]] + v[i])
  }
 }
}
\textbf{backtrack} using the final solution m and find the index of the data consumer \;
\textbf{return} List of selected data consumer \;
 \caption{Optimized privacy budget splitting algorithm}
\end{algorithm}

\begin{algorithm}
\SetAlFnt{\small}
\textbf{Input:} the data provider $\{u_1, \ldots, u_n\}$,the data consumer $\{D_1, \ldots, D_m\}$ \;
\textbf{Output:} List of the data provider and consumer pair that trade is completed  \;
\While{$i \leq m$}{
 $D_i$ calculate the parameter c to optimize the $f_i(\cdot$)\;
$D_i$ inform the $f_i(\cdot$) to the data provider
}
\While{$j \leq n$}{
$u_j$ report price $p_j$ to the data consumer
}
\While{$i \leq m$}{
\While{$j \leq n$}{
$D_i$ calculate the $\epsilon_j$ based on $p_j$\;
$D_i$ inform the $\epsilon_j$ to the $u_j$
}
}
\While{$j \leq n$}{
$u_j$ perform the \textbf{Optimized $\epsilon$ allocation algorithm} to maximize the utility
}

 \caption{ The proposed data trading process}
\end{algorithm}

\section{Experimental results}

In this section we   perform some experiments to  verify that the proposed  optimization method can find the best profit for the data consumer. 
For the experiments, we consider the families $\mathcal{F}$ of Examples~\ref{flog} and \ref{flinear}, namely  $\mathcal{F}=\{\ln(e-cp): c \in \mathbb{R}^+\}$ and $\mathcal{F}=\{1-cp: c \in \mathbb{R}^+\}$. For these two families 
the optimal parameter $c$ is also derived formally in  Appendix~\ref{app:b}.

The experimental variables are set as follows: we assume that there are $10$ data consumers, and the total number of data providers $n$ is set from $1000$ to $2000$ at an interval of $500$. The data provider's $\epsilon$ unit price is distributed normally with mean $1$ and standard deviation $1$, i.e., $\mathcal{N}(1,1)$, and convert $\epsilon$ unit price less than $0$ or more than $2$ to $0$ and $2$ respectively. We set the unit value $\epsilon$ to $0.1$, and the maximum $\epsilon$ value of data provider to $3$. 
We set the budgets as $60$, $90$, and $120$ and the number of the data provider as $1,000$, $1,500$, $2,000$. We assumed that the data consumer earned a profit of $10$ per $0.1$ epsilon and set the parameter $c$ to $1$ and $10$ for comparison. 

\begin{figure}[htbp]
\centerline{\includegraphics[width=0.82\textwidth]{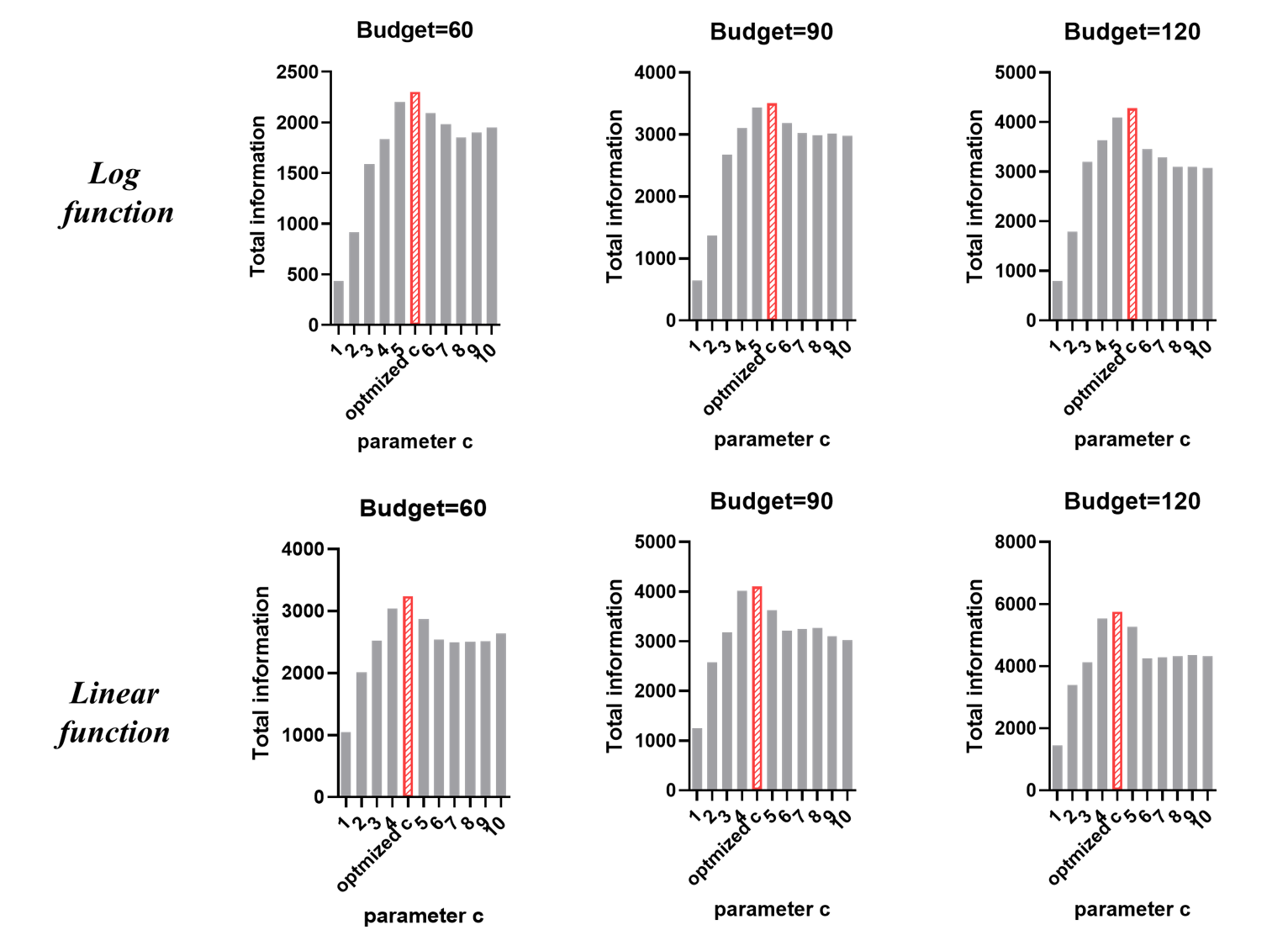}}
\caption{Experimental result of profit under a fixed budget. Log function is  the family $ln(e-cp)$ and Linear function is the family  $1-cp$. We let the parameter $c$ range from $0$ to $1$. The red bin represents the optimal 
value of $c$, namely the $c$ that gives maximum information.}
\label{fig7}
\end{figure}

The results are shown in Fig.~\ref{fig7}. For instance, in the case of the log family $\ln(e-cp)$, the optimal parameter $c$ is $5.36$, and in the case of the linear family $1-cp$, the optimal parameter $c$ is $4.9$. 
It is easy to verify that the optimal values of $c$ correspond to those 
determined by   solving Equations (8) and (13)  in Appendix~\ref{app:b}  of Examples~\ref{flog} and \ref{flinear}, respectively.


\section{Conclusion and future work}
As machine learning and data mining are getting  more and more deployed, individuals are becoming increasingly aware of the privacy issues and of the value of their data. This evolution of people's attitude towards privacy induces companies to develop new approaches to obtain personal data. We envision a scenario where data consumers can trade private data directly from the data provider by paying the price for the respective data, which has the potential to obtain personal information that could not be obtained in the traditional manner. In order to ensure a steady offer in  the data market, it is imperative to provide the privacy protection that the data providers deem necessary. Differential privacy can be applied to meet this requirement. However, the lack of standards for setting an appropriate value for the differential privacy parameter $\epsilon$, that determines the levels of data utility and privacy protection, makes it difficult to apply this framework in the data market. 

In order to address this problem, we have developed a method, based on incentives and optimization, to find an appropriate value for $\epsilon$ in the process of data trading. The proposed incentive mechanism motivates every data provider to report her privacy price honestly in order to maximize her benefit, and the proposed optimization method maximizes the profit for the data consumer under a fixed financial budget. Additionally, in an environment involving multiple data consumers, our mechanism suggests an optimal way for the data providers to split the privacy budgets, maximizing their utility. Through experiments, we have verified that the proposed method provides the best profits to the provider and consumer. 

Along the lines of what we have studied in this paper, there are many interesting research issues still open in this area. In  future work, we plan to study the following issues:
\begin{enumerate}
\item{ Mechanism for a fair incentive share in an environment where the data providers make a federation for privacy protection}
\item {Maximization of the data consumers' profits by estimating privacy price distribution of the data providers in an environment where demand of the data providers may change dynamically.}
\end{enumerate}

\section*{Acknowledgement}
This work was supported by the European Research Council (ERC) project HYPATIA under the European Union’s Horizon 2020 research and innovation programme. Grant agreement n. 835294.


\begin{subappendices}
\renewcommand{\thesection}{\Alph{section}}%

\section{Proofs}\label{app:a} 

\lemmagrtr*
\begin{proof}
Suppose the $i^{th}$ provider reports the privacy price as $p_i > \pi_i$. We want to show:
\begin{flalign}
\rho(\pi_i)-\rho(p_i)>0&&&\\\nonumber
&\\\nonumber\Leftrightarrow \mu(\pi_i)-\pi_if(\pi_i)-\mu(p_i)+\pi_if(p_i) >0&&\\\nonumber\nonumber
&\\\nonumber\Leftrightarrow \mu(\pi_i)-\mu(p_i)-\pi_i(f(\pi_i)-f(p_i))>0&&\\\nonumber
&\\\nonumber\Leftrightarrow \pi_if(\pi_i)+\int_{\pi_i}^{\infty}f(z)dz-(p_if(p_i)+\int_{p_i}^{\infty}f(z)dz)&&\\\nonumber
&\\\nonumber-\pi_i(f(\pi_i)-f(p_i))>0&&\\\nonumber
&\\\nonumber\Leftrightarrow \pi_if(p_i)-p_if(p_i)+\int_{\pi_i}^{p_i}f(z)dz>0&&\\\nonumber
&\\\nonumber\Leftrightarrow \int_{\pi_i}^{p_i}f(z)dz>f(p_i)(p_i-\pi_i)&&\\
\end{flalign}
Note that $f$ is a monotonically decreasing function. Let $g(p)=f(\pi_i)-f(p)$ for any $p\in\mathbb{R}$. 
Furthermore, $\int_{\pi_i}^{p_i}f(z)dz$ can be written as $f(p_i)(p_i-\pi_i)+\int_{\pi_i}^{p_i}g(z)dz$. Observe that $f$ is monotonically decreasing means $g(.)$ is monotonically decreasing, hence, $g(\pi_i)>g(p)$ for any $p>\pi_i$, $g(p_i)>0$. Hence, $\int_{\pi_i}^{p_i}g(z)dz>0 \Rightarrow \int_{\pi_i}^{p_i}f(z)dz>f(p_i)(p_i-\pi_i) \Rightarrow (2)$.
\qed
\end{proof}

\lemmasmll*
\begin{proof}
Similar and symmetric argument to Lemma \ref{reportedgrtr}.
\qed
\end{proof}

\thmequival*
\begin{proof}
Let $\tau$ be a monotonically increasing and additive function representing the pay-off earned by the data consumer by processing the information obtained from the different data providers. We wish to maximize her profit, i.e., $\sum_{i=1}^n(\tau(\epsilon_i)-\mu(p_i))$ for a fixed budget of expenses $B$, i.e., $\sum_{i=1}^n \mu(p_i)=B$. Therefore, we boil down to maximizing $\sum_{i=1}^n\tau(\epsilon_i)=\tau(\sum_{i=1}^n\epsilon_i)$ for  $\sum_{i=1}^n \mu(p_i)=B$, where $\epsilon_i=f(c,p_i)$. As $\tau$ is increasing, it attains maximum if and only if $\sum_{i=1}^n\epsilon_i=\sum_{i=1}^nf(c,p_i)$ is maximum. 
Then just observe that $\sum_{i=1}^nf(c,p_i)= I(\vb*{u})$ (cf. Definition~\ref{def:ti}). 
\end{proof}

\maxproffixedbudget*
\begin{proof}
The profit of the data consumer is $\sum_{i=1}^n(\tau(\epsilon_i)-\mu(p_i))$, where $\epsilon_i=f(c,p_i)$. Therefore, to maximize her profit for a fixed budget of expenses $B$ (used to pay the incentives to the data providers), i.e., $\sum_{i=1}^n\mu(p_i)=B$, we just have to maximize $\sum_{i=1}^n\tau(\epsilon_i)=\sum_{i=1}^n\tau(f(p_i))$, which by the assumption of additivity of $\tau$, is equal to $\tau(\sum_{i=1}^nf(c,p_i))$.

Note that the definition of $\mathcal{F}$ (cf. (\ref{eq:F})) ensures that 
all its function are continuous on  $c$.  Since the sum  of continuous functions is continuous, we have that also $\sum_{i=1}^nf(c,p_i)$ is continuous on $c$. Moreover, $c$ ranges in a closed interval: $c \in [0,c_{\text{min}}]$, where $c_{\text{min}}=\min_{p\in\{p_1,\ldots,p_n\}}\{c:f(c,p)=0\}$. Thus, by \textit{Extreme Value Theorem}, there exists a $c$ which maximizes $\sum_{i=1}^nf(c,p_i)$, which, in turn, maximizes $\sum_{i=1}^n(\tau(\epsilon_i)-\mu(p_i))$.

Furthermore, the condition of differentiability makes possible to apply the method of the lagrangians to find the maximum, 
by imposing that the partial derivatives are $0$. The  concavity condition implies that those points correspond to a (global) maximum.                              
\end{proof}

\maxinffixbudget*
\begin{proof}
By the sequential compositionality of differential privacy, the total information obtained by the data consumer is  $\sum_{i=1}^n\epsilon_i=\sum_{i=1}^nf(c,p_i)$. 
The rest follows as the proof of  Theorem~\ref{th:three}. \qed
\end{proof}

\section{Examples}\label{app:b} 
{\bf Example \ref{flog}}\;\;\;
Let $\mathcal{F}=\{\ln(e-cp): c \in \mathbb{R}^+\}$. 
We want to maximize 
\[F(c)=\sum_{i=1}^nf(c,p_i)\]
subject to the constraint \[G(c)=\sum_{i=1}^n\mu(p_i)=\sum_{i=1}^n[p_{i}\ln{(e-cp_i)}+\int_{p_{i}}^{\frac{e-1}{c}}\ln{(e-zp_i)}dz]-B=0\, ,
\]
for a fixed budget constant $B\in \mathbb{R^+}$. 
By the sequential compositionality of differential privacy, we have  $\sum_{i=1}^nf(c,p_i)=\sum_{i=1}^n\ln{(e-cp_i)}$. Furthermore,  $G(c)=\sum_{i=1}^n[p_{i}\ln{(e-cp_i)}+p_i\ln(\frac{e}{e-cp_i})+\frac{e}{c}\ln(e-cp_i)]-K=0$ where $K=B+\frac{n(e-1)}{c}$.  

Using the method of Lagrange multipliers, we define the Lagrangian function $\mathcal{L}(c,\lambda)=F(c)-\lambda(G(c)-K)=\sum_{i=1}^n\ln{(e-cp_i)}-\lambda(\sum_{i=1}^n[p_{i}\ln{(e-cp_i)}+p_i\ln(\frac{e}{e-cp_i})+\frac{e}{c}\ln(e-cp_i)]-K)=\sum_{i=1}^n\ln{(e-cp_i)}-\lambda(\sum_{i=1}^n[p_{i}+\frac{e}{c}\ln(e-cp_i)]-K)$. Thus we have
\begin{equation}
    \mathcal{L}(c,\lambda)=\sum_{i=1}^n\ln{(e-cp_i)}-\lambda\sum_{i=1}^np_{i}-\frac{\lambda e}{c}\sum_{i=1}^n\ln(e-cp_i)+\lambda K
\end{equation}

Hence, to find the optimal solution of $c$, we wish to solve the following: 
\begin{equation}
    \frac{\partial\mathcal{L}(c,\lambda)}{\partial\lambda}=0
    \label{partlamb}
\end{equation}

\begin{equation}
    \frac{\partial\mathcal{L}(c,\lambda)}{\partial c}=0
    \label{partc}
\end{equation}

Solving equation \eqref{partlamb}, we get
\begin{flalign}
\frac{\partial\mathcal{L}(c,\lambda)}{\partial \lambda}=0\nonumber&&&\\\nonumber
&\\\nonumber\Leftrightarrow K-\sum_{i=1}^n[p_{i}+\frac{e}{c}\ln(e-cp_i)]=0&&\\\nonumber\nonumber
&\\\nonumber\Leftrightarrow \sum_{i=1}^n[p_{i}+\frac{e}{c}\ln(e-cp_i)]=K&&\\\nonumber
&\\\nonumber\Leftrightarrow 
\ln(\prod_{i=1}^ne^{p_i}(e-cp_i)^{\frac{e}{c}})=K \label{implim}&&\\
\end{flalign}

Now,  observe  that $\lim_{c \to 0}\ln(\prod_{i=1}^ne^{p_i}(e-cp_i)^{\frac{e}{c}})=\infty$. Furthermore, observe that  $\lim_{c \to \infty}\ln(\prod_{i=1}^ne^{p_i}(e-cp_i)^{\frac{e}{c}})=0$. We recall that $K=B+\frac{n(e-1)}{c}$, hence $K$ is positive and finite. Thus, as $\ln(\prod_{i=1}^ne^{p_i}(e-cp_i)^e{\frac{e}{c}})$ is continuous, by the Intermediate Value Theorem \cite{ghosh2015}, we thereby can conclude that the curve  $y=\ln(\prod_{i=1}^ne^{p_i}(e-cp_i)^{\frac{e}{c}})$ intersects the curve $y=K$ at least once for $c\in(0,\infty)$, implying that we have at least one solution of \eqref{implim} ($\dagger_1$). 

Thereon, solving equation \eqref{partc}, we get 
\begin{flalign}
\frac{\partial\mathcal{L}(c,\lambda)}{\partial c}=0\nonumber&&&\\\nonumber
&\\\nonumber\Leftrightarrow 
\frac{\partial (\sum_{i=1}^n \ln(
e-cp_i)(1-\frac{\lambda e}{c})-\lambda\sum_{i=1}^np_i+K)}{\partial
c}=0&&\\\nonumber\nonumber
&\\\nonumber\Leftrightarrow 
\frac{\partial (\sum_{i=1}^n \ln(e-cp_i)(1-\frac{\lambda e}{c}))}{\partial c}=0&&\\\nonumber
&\\\nonumber\Leftrightarrow 
\sum_{i=1}^n\frac{-p_i}{e-cp_i}(1-\frac{\lambda e}{c})+\frac{\lambda e}{c^2}\sum_{i=1}^n\ln(e-cp_i)=0 \label{lambsol}&&\\
\end{flalign}
Equation \eqref{lambsol} is linear in $\lambda$, implying for every given $c$, we will find a $\lambda$ satisfying \eqref{lambsol} $(\dagger_2)$. 

Therefore, combining ($\dagger_1$) and ($\dagger_2$), we conclude that there is at least one optimal choice of $f(\cdot,\cdot)$ in $\mathcal{F}$ that maximizes the information gathered by the data consumer, subject to the fixed budget.\\

\noindent
{\bf Example \ref{flinear}} \;\;\; Let $\mathcal{F}=\{1-cp: c \in \mathbb{R}^+\}$.
We want to maximize 
\[
F(c)=\sum_{i=1}^nf(p_i)=\sum_{i=1}^n(1-cp_i)
\]
subject to 
\[
G(c)=\sum_{i=1}^n\mu(p_i)=\sum_{i=1}^n[p_{i}(1-cp_i)+\int_{p_{i}}^{\frac{1}{c}}(1-cz)dz]=B\,,
\]
for a fixed budget  $B\in \mathbb{R^+}$.

Observe that   $G(c)=\frac{n}{2c}-\frac{c}{2}\sum_{i=1}^np_i^2$.  
Using the method  of Lagrange multipliers, we define the Lagrangian function $\mathcal{L}(c,\lambda)=F(c)-\lambda(G(c)-B)=-c\sum_{i=1}^np_i-\lambda(\frac{n}{2c}-\frac{c}{2}\sum_{i=1}^np_i^2-B)$. Thus we have
\begin{equation}
    \mathcal{L}(c,\lambda)=-c\sum_{i=1}^np_i-\lambda(\frac{n}{2c}-\frac{c}{2}\sum_{i=1}^np_i^2-B)
\end{equation}

Hence, to find the optimal solution of $c$, we solve the following: 
\begin{equation}
    \frac{\partial\mathcal{L}(c,\lambda)}{\partial\lambda}=0
    \label{partlamblin}
\end{equation}

\begin{equation}
    \frac{\partial\mathcal{L}(c,\lambda)}{\partial c}=0
    \label{partclin}
\end{equation}

Solving equation \eqref{partlamblin}, we get
\begin{flalign}
\frac{\partial\mathcal{L}(c,\lambda)}{\partial \lambda}=0\nonumber&&&\\\nonumber
&\\\nonumber\Leftrightarrow -(\frac{n}{2c}-\frac{c}{2}\sum_{i=1}^np_i^2-B)=0&&\\\nonumber\nonumber
&\\\nonumber\Leftrightarrow \frac{n}{2c}-\frac{c}{2}\sum_{i=1}^np_i^2=B&&\\\nonumber
&\\\nonumber\Leftrightarrow 
c^2\sum_{i=1}^np_i^2+2Bc-n=0\label{implimlin}&&\\
\end{flalign}

Note that the LHS of equation \eqref{implimlin} is a quadratic equation in $c$, and as $4B^2+4n\sum_{i=1}^np_i^2>0$, equation \eqref{implimlin} has two distinct roots, or candidates for choosing the optimal $c$ $(\dagger_1')$.

Thereon, solving equation \eqref{partclin}, we get 
\begin{flalign}
\frac{\partial\mathcal{L}(c,\lambda)}{\partial c}=0\nonumber&&&\\\nonumber
&\\\nonumber\Leftrightarrow 
-\sum_{i=1}^np_i+\frac{\lambda n}{2c^2}+\frac{\lambda}{2}\sum_{i=1}^np_i^2=0 \label{lambsollin}&&\\
\end{flalign}
Equation \eqref{lambsollin} is linear in $\lambda$, implying for every given $c$, we will find a $\lambda$ satisfying \eqref{lambsollin} $(\dagger_2^')$. 

Therefore, combining ($\dagger_1^'$) and ($\dagger_2^'$), we conclude that there is at least one optimal choice of $f(\cdot,\cdot)$ in $\mathcal{F}_1$ that maximizes the information gathered by the data consumer, subject to the fixed budget.

\end{subappendices}
\end{document}